\documentclass[submission,copyright,creativecommons]{eptcs}
\providecommand{\event}{ICE '15} 
\usepackage{breakurl}

\usepackage{hyperref}
\usepackage{amsfonts}
\usepackage{amsmath}
\usepackage{amssymb}
\usepackage{amsthm}
\usepackage{amsbsy}
\usepackage{enumerate}
\usepackage{stackrel}
\usepackage{bm}
\usepackage{stmaryrd}
\usepackage{wasysym}
\usepackage{color}

\usepackage{derivation}
\usepackage{macros}

\usepackage{mathtools}
\usepackage{stmaryrd}

\usepackage{pifont}

\def\titlerunning{Orchestrated Session Compliance}
\def\authorrunning{Barbanera, van Bakel, de'Liguoro}

\title{Orchestrated Session Compliance
\thanks{This work was partially supported by COST Action IC1201 BETTY, MIUR PRIN
Project CINA Prot.\ 2010LHT4KM and Torino University/Compagnia San Paolo Project SALT.}
}

\author{Franco Barbanera
\institute{Dipartimento di Matematica e Informatica\\
University of Catania}
\email{barba@dmi.unict.it}
 \and 
Steffen van Bakel
\institute{Department of Computing \\
Imperial College London}
\email{svanbakel@imperial.ac.uk}
 \and 
Ugo de'Liguoro
\institute{Dipartimento di Informatica\\
University of Torino}
\email{ugo.deliguoro@unito.it}
}

\begin{document}

\setlength{\abovedisplayskip}{6pt}
\setlength{\belowdisplayskip}{\abovedisplayskip}

\maketitle

\begin{abstract}
We investigate the notion of orchestrated compliance for client/server interactions in the context of {\em session contracts}. 
Devising the notion of orchestrator in such a context makes it possible to have 
orchestrators with unbounded buffering capabilities and at the same time to guarantee any message from the client to be eventually delivered by the orchestrator to the server, 
while preventing the server from
sending messages which are kept indefinitely inside the orchestrator. The compliance relation is shown to be decidable by means of \emph{1}) a procedure synthesising the orchestrators, if any, making a client compliant with a server, and \emph{2}) a procedure for deciding whether an orchestrator behaves in a proper way as mentioned before.
\end{abstract}

\section{Introduction}
Session types and contracts are two formalisms used to study client/server protocols.
Session types have been introduced in \cite{honda.vasconcelos.kubo:language-primitives} as a tool for statically checking safe message exchanges through channels.
Contracts, on the other hand, as proposed in \cite{CCLP06,LP07,CGP10}, are a subset of CCS without $\tau$, that address the problem of abstractly describing behavioural properties of systems by means of process algebra.
In between these two formalisms lie {\em session contracts}\footnote{They were dubbed {\em session behaviours} in \cite{BdL10,BdL13}. For sake of uniformity and since {\em session contract} sounds more appealing, we adhere here to this name.} 
as introduced in~\cite{BdL10,BdL13,HennessyB12,BH13}; this is a formalism interpreting the session types into a subset of contracts.

In the theory of contracts, as well as in the formalism of session contracts, the notion of {\em compliance} plays a central role. 
A client $\rho$ is defined as being compliant with a server $\sigma$ (written as $\rho\comply\sigma$) whenever \emph{all} of its {\em requests} are satisfied by the server. 
Now it might be the case that client satisfaction cannot be achieved just because of a difference in the order in which the partners exchange information, or because one of them provide some extra un-needed information.

Consider the example of a meteorological data processing system (MDPS) that is permanently connected to a weather station to which it sends, according to its processing needs, particular data requests. 
For the sake of simplicity, we consider just two particular requests, namely for \emph{temperature} and \emph{humidity}. 
After the requests, the MDPS expects to receive the data in the order they were requested.
In the session-contracts formalism the interface for the simplified MDPS can be stated as follows:
 \[ \begin{array}{rcl}
\textsf{MDPS} &=& \rec x \procdot \Dual{\tt tempReq}\Actdot \Dual{\tt humReq}\Actdot {\tt temperature}\Actdot {\tt humidity}\Actdot x
 \end{array} \]
(Here, as in CCS, a symbol like `${\tt a}$' stands for on input action, whereas `${\tt \Dual{a}}$' denotes the corresponding output).
We assume a weather station to be able to send back the asked-for information in the order decided by its sensors, interspersed with information about \emph{wind speed}:
 \[ \begin{array}{rcl@{}l}
\textsf{WeatherStation} &=& \rec x \procdot {\tt tempReq} \Actdot {\tt humReq} \Actdot ( 
	& \Dual{\tt temperature} \Actdot \Dual{\tt humidity} \Actdot \Dual{\tt wind} \Actdot x \\
	&&& \oplus \\
	&&& \Dual{\tt humidity} \Actdot \Dual{\tt temperature} \Actdot \Dual{\tt wind} \Actdot x) 
 \end{array} \]

With the standard notion of compliance, it is not difficult to check that 
 $ \begin{array}{@{\,}r@{~}c@{~}l@{\,}}
\textsf{MDPS} &\not\comply& \textsf{WeatherStation}
 \end{array} $,
since the client \textsf{MDPS} has no input action for the wind data, and also since it could occur that the temperature and humidity data are delivered in 
a different order than expected by the \textsf{MDPS}.

A natural solution to this would consist of devising a process that acts as a mediator (here called \emph{orchestrator}) between the client and the server, coordinating them in a centralised way in order to make them compliant.
This sort of solution is the one adopted in the practice of web-service interaction, in particular for business processes, where the notion of orchestration has been introduced and developed:
 \begin{directquote}[\cite{OrchChor03}]
{\em Orchestration}:
Refers to an executable business process that may interact with both internal and external web services.
Orchestration describes how web services can interact at the message level, including
the business logic and execution order of the interactions. 
 \end{directquote}
In the context of the theory of contracts, this solution was formalised and investigated by Padovani \cite{Padovani10}, where orchestrators are processes that cannot affect the internal decisions of the client nor of the server, but can affect the way their synchronisation is carried out.

The orchestrating actions an orchestrator can perform have the following forms:
 \begin{description}

 \item[$\orchAct{\tt a}{\Dual{\tt a}}$ (resp.~$\orchAct{\Dual{\tt a}}{\tt a}$)]
the orchestrator gets ${\tt a}$ from the client (resp.~server) and immediately delivers it to the server (resp.~client) in a synchronous way.

 \item[$\orchAct{\tt a}{\varepsilon}$ (resp.~$\orchAct{\varepsilon}{\tt a}$)] 
the orchestrator gets ${\tt a}$ from the client (resp.~server) and stores it in the buffer.

 \item[$\orchAct{\Dual{\tt a}}{\varepsilon}$ (resp.~$\orchAct{\varepsilon}{\Dual{\tt a}}$)] the orchestrator takes ${\tt a}$ from the buffer and sends it to the client (resp.~server).

 \end{description}

So a possible orchestrator enabling compliance for our example would be
 \[ \begin{array}{rcl@{}l}
\textsf{f} &=& \rec x \procdot \orchAct{\tt tR}{\Dual{\tt tR}} \Actdot \orchAct{\tt hR}{\Dual{\tt hR}} \Actdot ( 
	& \orchAct{\Dual{\tt t}}{\tt t} \Actdot \orchAct{\Dual{\tt h}}{\tt h} \Actdot \orchAct{\varepsilon}{\tt w} \Actdot x \\
	&&& \vee \\
	&&& \orchAct{\varepsilon}{\tt h} \Actdot \orchAct{\Dual{\tt t}}{\tt t} \Actdot \orchAct{\Dual{\tt h}}{\varepsilon}\Actdot \orchAct{\varepsilon}{\tt w} \Actdot x) 
\end{array} \]
where {\tt tR}, {\tt hR}, {\tt t}, {\tt h}, and {\tt w} stand for {\tt tempReq}, {\tt humReq}, {\tt temperature}, {\tt humidity}, and {\tt wind}, respectively. 
The orchestrator \textsf{f} rearranges the order of messages when necessary, and {\em retains} the wind information, not needed by \textsf{MDPS}.

Actually, the orchestrator \textsf{f} is not a valid orchestrator in the sense of \cite{Padovani10}: indeed the ${\tt wind}$ information is never delivered to the client (i.e.~it is implicitly discarded), so that the buffer corresponding to \textsf{f} would be unbounded. 
Unbounded buffers are not allowed in \cite{Padovani10},
where boundedness of buffers is used to guarantee both decidability and the possibility of synthesising orchestrators. 
In a session setting instead, as is the present one, decidability and orchestrators synthesis can be established even in presence of unbounded buffering capabilities of orchestrators.

In a two-parties session-based interaction, the choice among several continuations always depends on just one of the two actors. To let our formalism fully adhere to such a viewpoint
our {\em session orchestrators}, besides (as argued in \cite{Padovani10}) being processes that cannot affect the internal decision of the client or the server, are such that they do not create any non-determinism besides that already present in the partners. 
This will correspond to restricting the syntax in such a way that orchestrators like, for instance,
$\orchAct {\varepsilon} {\tt a} \Actdot \textsf{f}_1 \vee \orchAct{\tt b}{\varepsilon} \Actdot \textsf{f}_2$, 
are not allowed. In fact, in the latter orchestrator, the choice of receiving an input from the client or from the server would not depend solely on the partners.
The \textsf{f} described above does respect this syntax restriction. 

Moreover, in our system it will be possible to prove that
 $
~~~\textsf{f} : \textsf{MDPS} \complyP \textsf{WeatherStation}~~~
 $
i.e.: \textsf{MDPS} and \textsf{WeatherStation} manage to be compliant (represented by $\complyP$ in our context) when their interaction is mediated by \textsf{f}.
In our system we will also manage to prevent the presence of \emph{fake orchestrated complying interactions}, like that between the client $\Dual{\tt a}\,{\tt .}\,\Dual{\tt b}$ and the server {\tt a} through the orchestrator $\orchAct{\tt a}{\Dual{\tt a}}\Actdot \orchAct{\tt b}{\varepsilon}$.
In this case the client gets the illusion that {\em all its requests} are satisfied, whereas its output {\tt b} never reaches the server, but will be indefinitely kept {\em inside} the orchestrator's buffer. 
While in the contract setting of \cite{Padovani10} such compliant interactions are allowed, in our session context we manage to rule out orchestrators behaving like $\orchAct{\tt a}{\tt \Dual{a}} \Actdot \orchAct{\tt b}{\varepsilon}$, which never deliver a message from the client to the server.

We shall prove that properties like the one just mentioned, characterising {\em well-behaved} orchestrators, are decidable.
Given an \textsf{f}, decidability of orchestrated compliance through \textsf{f} will be proved.
We will also show that, given a client and a server, it is possible to synthesise \emph{all} the orchestrators that make the client and system compliant, if any. 

\section{Session contracts and orchestrated compliance}
\label{sec:syntax}

{\em Session contracts}  are a restriction of {\em contracts} \cite{LP07,CGP10}. They are designed to be in one-to-one correspondence to session types \cite{honda.vasconcelos.kubo:language-primitives} without delegation (in \cite{BdL10,BdL13} a version with delegation was investigated). 
The restriction consists in constraining internal and external choices in a way that limits the non-determinism to (internal) output selection.

\begin{figure}[t]
\hrule
\vspace*{2mm}
\[ \begin{array}{lcl@{\hspace{4mm}}l}
\sigma,\rho & ::= & \stopA & \mbox{success} \\
 & \mid & a_1\Actdot\sigma_1 + \cdots + a_n\Actdot\sigma_n & \mbox{external choice} \\
 & \mid & \Dual{a}_1\Actdot\sigma_1 \oplus \cdots \oplus \Dual{a}_n\Actdot\sigma_n & \mbox{internal choice} \\
 & \mid & x & \mbox{variable} \\
 & \mid & \rec x \procdot  \sigma & \mbox{recursion}
\end{array} \]
\caption{The grammar of raw session-contracts}\label{fig:rbe-grammar}
\vspace*{2mm}
\hrule
\end{figure}

\begin{definition}[Session Contracts]\label{def:session-behaviours} 
\begin{enumerate}[i)]

 \item
Let $\Names$ be a countable set of symbols and~ $\CoNames = \Set{\Dual{a} \mid a \in \Names}$. 
The set $\Rbehav$ of {\em raw session contracts} is defined by the grammar in Figure~\ref{fig:rbe-grammar}, where:
\begin{itemize}

 \item 
for external and internal choices, $n \geq 1$, and $a_i \in\Names$ (hence $\Dual{a}_i\in\CoNames$) for all $1 \leq i\leq n$;
 
 \item 
the variable $x$ is a \emph{session-contract variable} out of a denumerable set; we consider occurrences of $x$ in $\sigma$ \emph{bound} in $\rec x \procdot  \sigma$.
An occurrence of $x$ in $\sigma$ is \emph{free} if it is not bound, and we write $\fv{\sigma}$ for the set of free variables in $\sigma$.
$\sigma$ is said to be {\em closed} whenever $\fv{\sigma}= \emptyset$.

 \end{itemize}

$\Act = \Names\cup\CoNames$ is the set of \emph{actions}.

\item
The set $\Sbehav$ of {\bf session contracts} is the subset of closed raw session contracts such that in \break\mbox{$a_1 \Actdot \sigma_1 + \cdots + a_n \Actdot \sigma_n$} ~and~ $\Dual{a}_1 \Actdot \sigma_1 \oplus \cdots \oplus \Dual{a}_n \Actdot \sigma_n$, the $a_i$ and the $\Dual{a}_i$ are, in both, pairwise distinct; moreover, in $\rec x \procdot \sigma$ the expression $\sigma$ is not a variable.

\end{enumerate}

\end{definition}

\noindent
As usual, we abbreviate $a_1\Actdot\sigma_1 + \cdots + a_n\Actdot\sigma_n$ by $\sum_{i=1}^n a_i\Actdot\sigma_i$, and $\Dual{a}_1\Actdot\sigma_1 \oplus \cdots \oplus \Dual{a}_n\Actdot\sigma_n$ by $ \bigoplus_{i=1}^n \Dual{a}_i\Actdot\sigma_i$. 
We also use the notations $\sum_{i\in I}a_i\Actdot\sigma_i$ and $\bigoplus_{i\in I} \Dual{a}_i\Actdot\sigma_i$ for finite and non-empty $I$.
We take the equi-recursive view of recursion, by equating $\rec x \procdot \sigma $ with $\sigma\Subst{x}{\rec x \procdot \sigma}$.

The trailing $\stopA$ is normally omitted: for example, we will write $a+b$ for $a \Actdot \stopA+ b \Actdot \stopA$. 
Session contracts will be considered modulo commutativity of internal and external choices.

The operational semantics of session contracts is given in terms of a labeled transition system (\LTS) 
$\sigma \lts{\alpha} \sigma'$ where $\sigma,\sigma'\in\Sbehav$ and $\alpha$ either belongs to a set of actions $\Act$ or is an internal action $\tau$.

\begin{definition}[\LTS\ for Session Contracts] \label{def:SB-red}~
We define the labelled transition system 
$(\Sbehav, \Act , \lts{})$ by 
 \[ \begin{array}{rcl@{\hspace{16mm}}rcl@{\hspace{16mm}}rcl}
\Dual{a}_1\Actdot\sigma_1 \oplus \cdots \oplus \Dual{a}_n\Actdot\sigma_n &\lts{\tau}& \Dual{a}_k\Actdot\sigma_k 
	&
\Dual{a}\Actdot\sigma &\lts{\Dual{a}}& \sigma 
	&
a_1\Actdot\sigma_1 + \cdots + a_n\Actdot\sigma_n &\lts{a_k}& \sigma_k
\end{array} \]
where $1 \leq k \leq n$, and $\sigma \lts{\alpha}\sigma'$ is short for $(\sigma,\alpha,\sigma') \in {\lts{}} $. We shall use $\lts{}$ as shorthand for $\lts{\tau}$.
As usual, we write $\Lts{} $ for $ \lts{}^*$ and $\Lts{\alpha} $ for $ \lts{}^* \lts{\alpha} \lts{}^*$ with $\alpha\in\Act$.
\end{definition}

Notice that reduction is not defined through contextual rules, so reduction only takes place at the `top' level.
Thereby, it is impossible for $\rec x \procdot  \sigma$ to unfold more than once without consuming a guard (remember that $\sigma$ is not a variable): so recursion is contractive in the usual sense.
We will safely assume that no two consecutive $\rec$ binders (as in $\rec x \procdot \rec y \procdot \sigma$) are present in a session contract. 

We observe that $\Lts{\alpha}$ is well defined, in that if $\sigma\in\Sbehav$ and $\sigma\Lts{\alpha}\sigma'$ (or $\sigma\Lts{}\sigma'$), then $\sigma'\in\Sbehav$.
\vspace{-6mm}
 \paragraph{Session orchestrators} 
As also done in \cite{Padovani10} in the context of the theory of contracts, we intend to investigate the notion of compliance when the interaction between a client and a server is mediated by an {\em orchestrator}.
Different from the broad contract setting, the session setting we are in induces some natural restrictions to the syntax of orchestrators, making it safe to have orchestrators with unbounded buffers. 
Moreover, it is possible to check whether any output from the client is eventually delivered by the orchestrator to the server, as well as whether there might be an infinite interaction which falsely progresses because it is made only of outputs from the server to the orchestrator
(see Section~\ref{sec:respdecidability}).

The set of actions an orchestrator can perform, that we take from \cite{Padovani10}, have been
informally described in the introduction.\footnote{One could wonder whether just asynchronous orchestration actions can be taken into account, since any $\orchAct{a}{\Dual{a}}$ action can be safely mimicked by two asynchronous ones, namely 
$\orchAct{a}{\varepsilon}.\orchAct{\varepsilon}{\Dual{a}}$ (similarly for $\orchAct{\Dual{a}}{a}$). A difference in fact would arise only for what concerns implementation, since the protocol for a 
synchronous exchange would not involve the use of a buffer, which is instead necessary for
asynchronous actions. Such an implementation issue seems unlikely to be related to our theoretical treatment. In contrast, we shall point out in Remark~\ref{rem:impl} how implementation related aspects might affect 
our formalisation.}

It can be reasonably argued that orchestrators must not show any internal non-determinism. 
Taking into account now the {\em session-based} interactions of our setting, such an assumption should be further extended, keeping in mind that in a session-based client/server interaction any possible non-determinism is due only to the internal non-determinism of the two partners. 
We therefore define our {\em session-orchestrators} so as to enforce this point of view.
It follows that the only choice we allow in session-orchestrators (represented by `$\vee$' in expressions like $f\vee g$) is an external one, and it is necessarily driven by the internal choice of one of the two partners. 
This implies that the actions immediately exhibited by $f$ and $g$ in an orchestrator like $f\vee g$ must have the same {\em direction}, i.e.~must belong to just one of the two subsets $\Set{ \orchAct{a}{\varepsilon}, \orchAct{a}{\Dual{a}}\mid a\in\Names } $ or $\Set{ \orchAct{\varepsilon}{a} \mid \orchAct{\Dual{a}}{a}\mid a\in\Names } $. 
Besides, orchestration actions of the form $\orchAct{\Dual{a}}{\varepsilon}$ or $\orchAct{\varepsilon}{\Dual{a}}$ must be used just as prefixes $\mu$ in orchestrators like $\mu.f$.
The other ruled-out cases, like $\orchAct{\Dual{c}}{\varepsilon}\Actdot f' \vee\orchAct{\varepsilon}{b}\Actdot g'$ or $\orchAct{c}{\varepsilon}\Actdot f' \vee\orchAct{\varepsilon}{b}\Actdot g'$, would conflict with the session viewpoint or, like $\orchAct{\Dual{c}}{\varepsilon}\Actdot f' \vee\orchAct{b}\varepsilon{}\Actdot g'$, would be meaningless according to the syntax of session contracts.

We now formally define orchestration actions by partitioning them into different syntactic categories.

\begin{definition}[Session-orchestration actions]
We define {\em $\OrchAct$} as the set of {\em session-orchestration actions} $\mu$ described by the following grammar (where $a\in\Names$ and $\Dual{a}\in\CoNames$): 
 \[ \begin{array}{lcl@{\hspace{18mm}}l lcl@{\hspace{4mm}}l }
\mu & ::= & \iota_L \mid \iota_R \mid o &
o & ::= &\orchAct{\Dual{a}}{\varepsilon} \mid \orchAct{\varepsilon}{\Dual{a}} \\
\iota_L & ::= &\orchAct{a}{\varepsilon} \mid \orchAct{a}{\Dual{a}} &
 \iota_R & ::= &\orchAct{\varepsilon}{a} \mid \orchAct{\Dual{a}}{a} 
\end{array} \]
\end{definition}
We let $\mu, \mu', \mu_1,\ldots$ range over orchestration actions,
and $\vec{\mu}$ over both finite sequence $\mu_1\ldots\mu_n$ in $\OrchAct^*$ and
infinite sequence $\mu_1\ldots\mu_n\ldots$ in $\OrchAct^\infty$.

\begin{definition}[Session Orchestrators]
We define $\Orch$ as the set of {\em session orchestrators}, ranged over by $f, g, h,\ldots$, described by the {\em closed} terms generated by the following grammar: 
 \[ \begin{array}{lrl@{\hspace{4mm}}l}
f, g & ::= & \stopf \\
 & \mid & \iota_L\Actdot f_1 \vee\cdots\vee \iota_L\Actdot f_n & (n\geq 1)\\
 & \mid & \iota_R\Actdot f_1 \vee\cdots\vee \iota_R\Actdot f_n & (n\geq 1)\\
 & \mid & o\Actdot f \\
 & \mid & x \\
 & \mid & \rec x \procdot f 
\end{array} \] 
We impose session orchestrators to be {\em contractive}, i.e.~the $f$ in $\rec x \procdot f$ is assumed to not be a variable.
\end{definition}

The expression
$\stopf$ represents the orchestrator offering no action. $o\Actdot f$ offers just the orchestration action of the category $o$ and
continues as $f$, whereas {$\iota_L\Actdot f_1 \vee\cdots\vee \iota_L\Actdot f_n$} and 
$\iota_R\Actdot f_1 \vee\cdots\vee \iota_R\Actdot f_n$ offer $n$ (uni-directional) actions of the syntactical categories, respectively, $\iota_L$ and $\iota_R$.
Recursive orchestrators can be expressed by means of the $\rec$ binder and recursion variables,
in the usual way. As for session contracts, orchestrators are defined as to have recursion
variables guarded by at least one orchestration action. 
In the following we shall often refer to `session orchestrators' as simply `orchestrators.'
As for session contracts, we take an equi-recursive point of view, so identify
$\rec x \procdot f$ and $f\Subst{x}{\rec x \procdot f}$.

We now define the operational semantics of orchestrators as an \LTS.

\begin{definition}[\LTS\ for Orchestrators] 
We define the labelled transition system 
$(\Orch, \mbox{\em $\OrchAct$} , \lts{})$ by 
 \[ \begin{array}{c@{\hspace{16mm}}c@{\hspace{16mm}}c}
\Inf{ \mu \Actdot f \OrchStep{\mu} f }
 & 
\Inf{
 f \OrchStep{\mu} f'
}{
 f\vee g \OrchStep{\mu} f'
}
 &
\Inf{
 g \OrchStep{\mu} g'
}{
 f\vee g \OrchStep{\mu} g'
}
\end{array} \]
Given a sequence $\vec{\mu}$, we write $f \OrchStep{\vec{\mu}} {} $ whenever
$f \OrchStep{\mu_1} f_1 \OrchStep{\mu_2} \cdots \OrchStep{\mu_n} f_n$ if $\vec{\mu}=\mu_1\cdots\mu_n\in\mbox{\em $\OrchAct$}^*$, or
$f \OrchStep{\mu_1} \cdots \OrchStep{\mu_n} f_n \!\OrchStep{\mu_{n+1}}\! \cdots\,$ if $\vec{\mu}=\mu_1\cdots\mu_n\cdots\in\mbox{\em $\OrchAct$}^\infty$. 
We write $f \noOrchStep $ if there is no $\mu$ such that $f \OrchStep{\mu} $. 
\end{definition}

\begin{definition}[Orchestrator Traces] \label{def:traces}
Let $f\in\Orch$.
\begin{enumerate} 
\item
The set $\Trace{f}\subseteq (\mbox{\em $\OrchAct$}^*\cup\mbox{\em $\OrchAct$}^\infty)$ of \ {\em traces of} $f$ is defined by:
$\Trace{f} = \Set{ \vec{\mu} \mid f \OrchStep{\vec{\mu}} {} } $.

\item
The set $\MaxTrace{f}\subseteq (\mbox{\em $\OrchAct$}^*\cup\mbox{\em $\OrchAct$}^\infty)$ of \ {\em maximal traces of}
 $f$ is defined by 
 \[ \begin{array}{rcl}
\MaxTrace{f} &=& \Set{ \vec{\mu}\in\Trace{f} \mid \exists f' \Pred[ f \OrchStep{\vec{\mu}}f' \noOrchStep ]
\mbox{ or } \mu\in\mbox{\em $\OrchAct$}^\infty } 
 \end{array} \]
 
\end{enumerate}
\end{definition}

As in \cite{Padovani10}, we define an {\em orchestrated system} as a triple $\langle \rho, f, \sigma \rangle$ (written $\rho \pf{f} \sigma$) representing $\rho$ (the client) and $\sigma$ (the server)  interacting with each
other under  the supervision of $f$. 


\begin{definition}[Orchestrated Systems operational semantics]
The operational semantics of orchestrated systems is defined as follows:
\[ \begin{array}{c@{\hspace{24mm}}c}
\Inf{
 \rho\lts{}\rho'
}{
 \rho \pf{f} \sigma \lts{} \rho' \pf{f} \sigma
}
&
\Inf{
 \sigma\lts{}\sigma'
}{
\rho \pf{f} \sigma \lts{} \rho \pf{f} \sigma'
}
\end{array} \] 
\[ \begin{array}{c@{\hspace{14mm}}c@{\hspace{14mm}}c}
\Inf{
 \rho \lts{\alpha}\rho' \quad f\ltsOrch{\Dual{\alpha}}{\alpha} f' \quad \sigma \lts{\Dual{\alpha}}\sigma'
}{
 \rho \pf{f} \sigma \lts{\langle{\Dual{\alpha}},{\alpha}\rangle} \rho' \pf{f'} \sigma'
} 
&
\Inf{
 \rho \lts{\Dual{\alpha}}\rho' \quad f\ltsOrch{\alpha}{\varepsilon} f'
}{
 \rho \pf{f} \sigma \lts{\langle{\alpha},{\varepsilon}\rangle} \rho' \pf{f'} \sigma
}
&
\Inf{
 f \ltsOrch{\varepsilon}{\alpha} f' \quad \sigma \lts{\Dual{\alpha}}\sigma' 
}{
 \rho \pf{f} \sigma \lts{\langle{\varepsilon},{\alpha}\rangle} \rho \pf{f'} \sigma'
}
\end{array} \]



\noindent
We write $\Lts{\mu}$ for $\rtclts{} \comps \lts{\mu} \comps \rtclts{} $, and $\Lts{\vec{\mu}}$ for $\Lts{\mu_1} \comps \cdots \comps \Lts{\mu_n}$ (resp.~$\Lts{\mu_1} \comps \Lts{\mu_2} \comps \cdots$) if $\vec{\mu}$ is finite (resp.~infinite). 
The notation 
$ \rho\pb{f}\sigma \notlts {} $ will be used when both  $ \rho\pb{f}\sigma\notlts{}$(according to the first two rules above) and $\neg \exists \mu \Pred[ \rho\pb{f}\sigma\Lts{\mu}{}]$ hold.

\end{definition}

Notice that for the operational semantics of orchestrated systems we have defined labelled reductions instead of a reduction 
relation (as done in \cite{Padovani10}). We label orchestrated-systems' transitions by the 
orchestration actions which make them possible, since in our setting we need to check for particular conditions of orchestrator buffers after the evolution of an orchestrated system. A buffer can be explicitly coupled with an orchestrator or can be represented implicitly by the actions performed by the orchestrator. The latter is the choice of \cite{Padovani10}, that we maintain. 

We now define a notion of compliance which is coarser than expected because of possible unfair behaviour of the orchestrators, which will be refined in Section~\ref{sec:respdecidability}.

\begin{definition}[Disrespectful and Strict Orchestrated Compliance]
\label{def:disstrictcompl}
An orchestrator $f$ is said to be $\rho\Strictdot \sigma$ {\em strict} whenever, for any finite $\vec{\mu}$, $f\OrchStep{\vec{\mu}} $ implies $ \rho\|_f\sigma \Lts{\vec{\mu}} {}$. 
We define:
\begin{enumerate}[i)]
\item
\label{def:disstrictcompl-i}
$f: \rho\complyPds \sigma$ ~if $f$ is $\rho \Strictdot  \sigma$ strict, and for any $\vec{\mu}$, $\rho'$ and $\sigma'$, the following holds:
\[ \begin{array}{rcl}
\rho \pf{f} \sigma \Lts{\vec{\mu}} \rho' \pf{f'} \sigma' \notlts{} &\mbox{implies}& \rho'=\mbox{\em $\stopA$}.
 \end{array} \]
\item
$ \begin{array}{@{}rcl}
\rho\complyPds \sigma &\textrm{if}& \exists f \Pred[ f: \rho\complyPds \sigma].
 \end{array} $

\end{enumerate}
\end{definition}

\section{Orchestrators Synthesis} 
\label{sec:orchsynth}

In this section we define an inference system $\derinf$ for (possibly open) orchestrators, deducing judgments like \mbox{$f: \rho \complyPdsF \sigma$}, under finitely many assumptions of a certain shape.
We first establish that the system is sound with respect to the $\complyPds$ relation. 
Then, on the basis of that system, we provide an algorithm $\Synth$ for orchestrator synthesis which, given $\rho$ and $\sigma$, returns the set of all the relevant orchestrators $f$ such that $\derinf f: \rho \complyPdsF \sigma$ (namely with $\Gamma = \emptyset$) and hence that $ f : \rho \complyPds \sigma$. 
The algorithm is essentially an exhaustive proof search for $\derinf$ that can be shown to be always terminating. 

\begin{definition}[The orchestrators inference system $\derinf$] 
\label{def:inferenceSyst}
The judgements of the system are expressions of the form 
$\Gamma\derinf f: \rho \complyPdsF \sigma$, where
 $\rho,\sigma\in\Sbehav$, $f$ is a (possibly open) orchestrator and 
$\Gamma$ is a set of assumptions of the form $x:\rho_i \complyPdsF \sigma_i$ such that: 
$x{:}\rho \complyPdsF \sigma \in\Gamma \And  y{:}\rho \complyPdsF \sigma\in\Gamma \implies x=y$
~(so $\Gamma$ represents an injective mapping from variables to expressions of the form $\rho \complyPdsF \sigma$). 
The axioms and rules of the system are described in Figure~\ref{fig:infsys}.
\end{definition}

In the inference system of Figure~\ref{fig:infsys} the symbol $\complyPdsF$ is a relation symbol
representing the relation $\complyPds$ as defined in Definition~\ref{def:disstrictcompl}.
In order to give the intuition behind the inference system, let us briefly comment on one of the rules, say $(\TcomplSumL)$. 
In case it is possible to show that $f'$ is an orchestrator for
$f_p\complyPds \sigma$, orchestrated compliance can be obtained for $\sum_{i\in I} a_i\Actdot {\rho}_i \complyPds \sigma$ by means of $\orchAct{\Dual{a}_p}{\varepsilon}\Actdot f'$, since the $\orchAct{\Dual{a}_p}{\varepsilon}$ action satisfies one of the {\em input requests} $a_i$s.
In case $x\not\in fn(f')$, we get that $\rec x \procdot \orchAct{\Dual{a}_p}{\varepsilon}\Actdot f' = \orchAct{\Dual{a}_p}{\varepsilon}\Actdot f'$. This means that axiom $(\TcomplAx)$ has been used in
the derivation of $f'$ and the interaction between $\sum_{i\in I} a_i\Actdot {\rho}_i$ and $\sigma$
finitely succeeds if the actions described in the branch from $(\TcomplSumL)$ to $(\TcomplAx)$ are performed. In case  $x\in fn(f')$, rule $(\TcomplHyp)$ has been used for $f'$, and a successful infinite interaction is possible between $\sum_{i\in I} a_i\Actdot {\rho}_i$ and $\sigma$
when the orchestrator  repeatedly performs the actions in the branch from $(\TcomplSumL)$ to $(\TcomplHyp)$, as described by the recursive orchestrator $\rec x \procdot \orchAct{\Dual{a}_p}{\varepsilon}\Actdot f'$.

\begin{figure}[t]
\hrule 
 \[ \begin{array}{@{\!}rl@{\quad}rl}
(\TcomplAx):&
\Inf { 
\Gamma\derinf \stopf : \stopA \complyPdsF \sigma
} &
(\TcomplHyp): & 
\Inf { 
\Gamma, x{:}\rho\complyPdsF\sigma \derinf x : \rho\complyPdsF\sigma 
}
	\\ [5mm]
(\TcomplSumL): &
\Inf [p\in I]{ 
	\Gamma,\; x{:} \sum_{i\in I} a_i\Actdot {\rho}_i \complyPdsF \sigma \derinf f' : {\rho}_p \complyPdsF \sigma
 }{ 
	\Gamma\derinf \rec x \procdot \orchAct{\Dual{a}_p}{\varepsilon}\Actdot f' : \sum_{i\in I} a_i\Actdot {\rho}_i \complyPdsF \sigma
} 
	&
(\TcomplSumR): & 
\Inf [p\in I]{ 
		\Gamma,\; x{:} \rho\complyPdsF \sum_{i\in I} a_i\Actdot {\sigma}_i \derinf f' : \rho \complyPdsF \sigma_p
 }{ 
	\Gamma\derinf \rec x \procdot \orchAct{\varepsilon}{\Dual{a}_p}\Actdot f' : \rho\complyPdsF \sum_{i\in I} a_i\Actdot {\sigma}_i
} 	
	\\ [6mm]
(\TcomplOOA): & \multicolumn{3}{l}{
\Inf { 
	\Gamma,\; x{:} \bigoplus_{i\in I} \Dual{a}_i\Actdot {\rho}_i \complyPdsF \bigoplus_{j\in J} \Dual{b}_j\Actdot {\sigma}_j \derinf f_j : \bigoplus_{i\in I} \Dual{a}_i\Actdot {\rho}_i\complyPdsF \sigma_j \quad (\forall j\in J)
 }{ 
	\Gamma\derinf \rec x \procdot  \bigvee_{j\in J} \orchAct{\varepsilon}{b_j}\Actdot f_j : \bigoplus_{i\in I} \Dual{a}_i\Actdot {\rho}_i \complyPdsF \bigoplus_{j\in J} \Dual{b}_j\Actdot {\sigma}_i
} 
}	\\ [6mm]
(\TcomplOOB): & \multicolumn{3}{l}{
\Inf { 
	\Gamma,\; x{:} \bigoplus_{i\in I} \Dual{a}_i\Actdot {\rho}_i \complyPdsF \bigoplus_{j\in J} \Dual{b}_j\Actdot {\sigma}_i \derinf f_i : {\rho}_i \complyPdsF \bigoplus_{j\in J} \Dual{b}_j\Actdot {\sigma}_i \quad (\forall i\in I)
 }{ 
	\Gamma\derinf \rec x \procdot  \bigvee_{i\in I} \orchAct{a_i}{\varepsilon}\Actdot f_i: \bigoplus_{i\in I} \Dual{a}_i\Actdot {\rho}_i \complyPdsF \bigoplus_{j\in J} \Dual{b}_j\Actdot {\sigma}_i
} 
}	\\ [6mm]
(\TcomplOSum): & \multicolumn{3}{l}{
\Inf [I=H\cup K, K\subseteq J]{ 
	\Gamma'\derinf f_i : \rho_i\complyPdsF\sum_{j\in J} a_j\Actdot {\sigma}_j \quad(\forall i \in H) \qquad 	\Gamma'\derinf f_i : {\rho}_i \complyPdsF\sigma_i \quad (\forall i \in K)
 }{ 
	\Gamma\derinf \rec x \procdot (\bigvee_{h\in H} \orchAct{a_h}{\varepsilon}\Actdot f_h)\vee (\bigvee_{k\in K} \orchAct{a_k}{\Dual{a}_k}\Actdot f_k) : \bigoplus_{i\in I} \Dual{a}_i\Actdot {\rho}_i \complyPdsF \sum_{j\in J} a_j\Actdot {\sigma}_j
} 
}	\\ [3mm]
\textrm{where} & \Gamma' = \Gamma,\;x{:}\bigoplus_{i\in I} \Dual{a}_i\Actdot {\rho}_i \complyPdsF \sum_{j\in J} a_j\Actdot {\sigma}_j.
	\\ [4mm]
(\TcomplSumO): & \multicolumn{3}{l}{
\Inf [J=H\cup K, K\subseteq I]{ 
	\Gamma'\derinf f_j : \sum_{i\in I} a_i\Actdot {\rho}_i \complyPdsF\sigma_j \quad (\forall j \in H) \qquad
	\Gamma'\derinf f_j : \rho_j\complyPdsF{\sigma}_j \quad (\forall j \in K )
 }{ 
	\Gamma\derinf \rec x \procdot (\bigvee_{h\in H} \orchAct{\varepsilon}{a_h}\Actdot f_h)\vee (\bigvee_{k\in K} \orchAct{\Dual{a}_k}{a_k}\Actdot f_k) : \sum_{i\in I} a_i\Actdot {\rho}_i \complyPdsF \bigoplus_{j\in J} \Dual{a}_j\Actdot {\sigma}_j
} 
}	\\ [3mm]
\textrm{where} & \Gamma' = \Gamma,\; x{:} \sum_{i\in I} \Dual{a}_i\Actdot {\rho}_i \complyPdsF \bigoplus_{j\in J} a_j\Actdot {\sigma}_j
 \end{array} \]
\caption{The inference system $\derinf$.}\label{fig:infsys}
\vspace*{2mm}
\hrule 
\end{figure}


\begin{definition}[Judgment Semantics]
\label{def:ComplModel}
Let $\Gamma = \Set{x_1{:}\rho_1\complyPdsF\sigma_1,\ldots, x_k{:}\rho_k\complyPdsF\sigma_k}$, and $\theta$ be a map such that
$\theta(x_i) = f_i$, where the $f_i$s are proper (i.e.~closed) orchestrators. Then we define:
 \[ \begin{array}{lcl}
\theta \models \Gamma &\ByDef& 
\forall {(x_i{:}\rho_i\complyPdsF\sigma_i) \in \Gamma} \Pred[ \theta(x_i):\rho_i\complyPds\sigma_i ]
	\\
\Gamma \models f:\rho \complyPdsF\sigma &\ByDef& 
\forall \theta \Pred[ \theta\models \Gamma \implies \theta(f):\rho\complyPds\sigma ]
 \end{array} \]
where $ \theta(f)$ is the result of substituting, for all variables $x\in f$, all free occurrences of $x$ by $\theta(x)$.
\end{definition}

\begin{theorem}[Soundness]\label{thm:soundeness-derinf}
If ~$\Gamma \derinf f :\rho\complyPdsF \sigma$~ then ~$\Gamma \models f:\rho \complyPdsF\sigma$.
\end{theorem}
\begin{proof} (Sketch)
It is possible to device a sound and complete system $\der$ for judgments of the shape $\Gamma \der f :\rho\complyPdsF \sigma$, where $f$ is a closed orchestrator and where
$\Gamma$ is a set of assumptions on closed orchestrators (not on variables as in $\derinf$). Now
it can be proved that if $f$ is closed and $\Gamma \derinf f :\rho\complyPdsF\sigma$ 
is derivable, then for any $\theta$ such that
$\theta\models\Gamma$ we have $\theta(\Gamma)\der f :\rho\complyPdsF\sigma$, where $\theta(\Gamma)$ is the result of substituting all orchestrator variables $x$ by $\theta(x)$. Then the thesis follows from the soundness of $\der$.
\end{proof}

\begin{figure}[t]
\hrule
\vspace{2mm}
{\small
\begin{tabbing}
\Synth\=$(\Gamma, \rho,\sigma)$ = \+ \\ [2mm]
\IF\ $x:\rho\complyPdsF\sigma \in \Gamma$ \THEN\ $\Set{x}$ 
	\\ [2mm]
\ELSE\; \= \IF\ $\rho = \stopA$ \THEN\ $\Set{\stopf}$ 
	\\ [2mm]
\ELSE \> \IF\ \= $\rho = \sum_{i\in I} a_i\Actdot {\rho}_i $ \AND\ $\sigma = \sum_{j\in J} a_j\Actdot {\sigma}_j $ \THEN \\ 
	\> \LET\ \= $\Gamma' = \Gamma,\; x{:}\rho\complyPdsF \sigma$ \IN \\ 
	\>\> $\bigcup_{i\in I}\Set{\rec x \procdot \orchAct{\Dual{a}_i}{\varepsilon}\Actdot f \mid f\in \Synth\,(\Gamma',{\rho}_i ,\sigma) } ~ \cup~ \bigcup_{j\in J}\Set{\rec x \procdot \orchAct{\varepsilon}{\Dual{a}_j}\Actdot f \mid f\in \Synth\,(\Gamma', {\rho}, \sigma_j) }$ 
	\\ [2mm]
\ELSE \> \IF\ \= $\rho = \bigoplus_{i\in I} \Dual{a}_i\Actdot {\rho}_i $ \AND\ $\sigma = \bigoplus_{j\in J} \Dual{a}_j\Actdot {\sigma}_i$ \THEN \\ 
	 \> \LET\ \= $\Gamma' = \Gamma,\; x{:}\rho\complyPdsF \sigma$ \IN \\ 
\>\>$\Set{\rec x \procdot \bigvee_{i\in I} \orchAct{a_i}{\varepsilon}\Actdot f_i \mid f_i\in \Synth\,(\Gamma', {\rho}_i , \sigma) }\,\cup\,\Set{\rec x \procdot \bigvee_{j\in J} \orchAct{\varepsilon}{a_j}\Actdot f_j \mid f_j\in \Synth\,(\Gamma', {\rho} , \sigma_j) }$ 	\\ [2mm]

\ELSE \> \IF\ \= $\rho = \bigoplus_{i\in I} \Dual{a}_i\Actdot {\rho}_i $ \AND\ $\sigma = \sum_{j\in J} a_j\Actdot {\sigma}_j$ 
			 \THEN\ \\ 
	 \> \LET\ \= $\Gamma' = \Gamma,\; x{:}\rho \complyPdsF \sigma$ \IN \\ 
\>\> $\Set{\rec x \procdot  (\bigvee_{h\in H} \orchAct{a_h}{\varepsilon}\Actdot f_h)
~\vee~(\bigvee_{k\in K} \orchAct{a_k}{\Dual{a}_k}\Actdot f_k)$ 
	\\ \>\> \qquad
$ \mid I= H\cup K, K\subseteq J, f_h\in \Synth(\Gamma', \rho_h, \sigma), f_k\in\Synth\,(\Gamma', {\rho}_k ,\sigma_k ) }$ \\ 
\>\> $\cup ~ \bigcup_{j\in J}\Set{\rec x \procdot \orchAct{\varepsilon}{\Dual{a}_j}\Actdot f \mid f\in \Synth\,(\Gamma', {\rho} , \sigma_j }$ 
	\\ [2mm]

\ELSE \> \IF\ \= \ $\rho = \sum_{i\in I} a_i\Actdot {\rho}_i $ \AND\ $\sigma = \bigoplus_{j\in J} \Dual{a}_j\Actdot {\sigma}_j$
			 \THEN\ \\ 
	 \> \LET\ \= $\Gamma' = \Gamma,\;x{:}\rho \complyPdsF \sigma$ \IN \\ 
\>\>$\Set{\rec x \procdot  (\bigvee_{h\in H} \orchAct{\varepsilon}{a_h}\Actdot f_h)
~\vee~(\bigvee_{k\in K} \orchAct{\Dual{a}_k}{a_k}\Actdot f_k) $ 
	\\ \>\> \qquad
$ \mid J = H\cup K, K\subseteq I, f_h\in \Synth((\Gamma', \rho,\sigma_h), f_k\in\Synth\,(\Gamma', {\rho}_k,\sigma_k )}$ \\ 
\>\> $\cup ~ \bigcup_{i\in I}\Set{\rec x \procdot \orchAct{\Dual{a}_i}{\varepsilon}\Actdot f \mid f\in \Synth\,(\Gamma', {\rho_i} , \sigma)}$ 
	\\ [2mm]
\ELSE \> $\emptyset$
\end{tabbing}
\vspace{-4mm}
}\caption{The algorithm \Synth.}\label{fig:Synth}
\vspace{2mm}
\hrule

\end{figure}

The synthesis algorithm $\Synth$ is defined in Figure~\ref{fig:Synth}. Given a set of assumptions $\Gamma$, a client $\rho$ and a server $\sigma$, the
algorithm computes a set of orchestrators $\cal F$ such that for all $f \in {\cal F}$ a derivation of $\Gamma\derinf f:\rho\complyPdsF\sigma$ exists.
The algorithm essentially mimics
the rules of the inference system of Figure~\ref{fig:infsys}. Intuitively, in case we are looking 
for orchestrators for
$\rho = \bigoplus_{i\in I} \Dual{a}_i\Actdot {\rho}_i $ and $\sigma = \bigoplus_{j\in J} \Dual{a}_j\Actdot {\sigma}_i$ under the assumptions $\Gamma$, we notice that they can be inferred for such
$\rho$ and $\sigma$ in system $\derinf$ only by means of rules 
$(\TcomplOOA)$ or $(\TcomplOOB)$ and that their form is, respectively,
 $\rec x \procdot \bigvee_{i\in I} \orchAct{a_i}{\varepsilon}\Actdot f_i$ or
$\rec x \procdot \bigvee_{j\in J} \orchAct{\varepsilon}{a_j}\Actdot f_j$, where the $f_i$s and the $f_j$s are the orchestrators for the pairs $\rho_i$,$\sigma$ and $\rho$,$\sigma_j$, respectively. This accounts 
for the fourth clause in the synthesis algorithm. We can prove the algorithm to be sound.

\begin{lemma}\label{lem:soundness-Synth-a}
If {\em \Synth}$(\Gamma, \rho, \sigma) = {\cal F} \neq \, \emptyset$ then, for all $f\in {\cal F}$, $\Gamma\derinf f:\rho\complyPdsF\sigma$ is derivable.
\end{lemma}

On the other hand, the algorithm is complete in the following sense:

\begin{lemma}\label{lem:soundness-Synth-b}
If $f : \rho\complyPds \sigma$ and {\em \Synth}$(\emptyset,\rho,\sigma)$ terminates then there exists $g$ such that $g \in \mbox{\em \Synth}(\emptyset,\rho,\sigma)$.
\end{lemma}
The $g$ of the above lemma {\em represents} $f$. In particular it could be got by  {\em delaying} the termination of the algorithm when the first clause is applicable. Moreover, it could
be got out of $f$ by replacing syncronous actions by pairs of asyncronous ones (or also by simply adding asyncronous actions). For instance,
for $\rho=\rec x. \Dual{a}.x$ and $\sigma=\rec x. a.x$ the orchestrator
$f=\orchAct{a}{\varepsilon}.\orchAct{\varepsilon}{\Dual{a}}.\rec x. \orchAct{a}{\Dual{a}}.x$
correctly mediates between $\rho$ and $\sigma$, the algorithm terminates, but $f\not\in\mbox{ \Synth}(\emptyset,\rho,\sigma)$. On the other hand, $g=\rec x. \orchAct{a}{\Dual{a}}.x$ belongs to $\mbox{\Synth}(\emptyset,\rho,\sigma)$ and it is related to $f$ in the sense above. It can be shown, besides, that if $f$ is {\em respectful} in the sense of Sect. \ref{sec:respdecidability} below, there exists a respectful $g$ in $\mbox{\Synth}(\emptyset,\rho,\sigma)$.

It remains to show that $\Synth$ is terminating:

\begin{lemma}\label{lem:termination-Synch}
For all $\Gamma$, $\rho$ and $\sigma$, the execution of {\em \Synth}$\,(\Gamma,\rho, \sigma)$ terminates.
\end{lemma}

\begin{proof} (Sketch)
The proof is based on the fact that all session contracts in the recursive calls of $\Synth$ are a sub-expression of either $\rho$ or $\sigma$ or of a session contract in a judgment in $\Gamma$ (which is finite). Since session contracts are regular trees, their sub-expressions are a finite set, so that the test $x{:}\rho\complyPdsF\sigma \in \Gamma$ (where $x$ is any variable) at the beginning of $\Synth$ cannot fail infinitely many times.
\end{proof}

\begin{corollary}\label{cor:decidable-complPds}
The relation $\complyPds$ is decidable; moreover if $\rho \complyPds \sigma$ then it is possible to compute a set ${\cal F}$ of orchestrators for $\rho$ and $\sigma$
\end{corollary}
Recall that the computed orchestrators {\em represent} all the possible orchestrators, in the sense of the discussion after Lemma \ref{lem:soundness-Synth-b}.



\section{Respectfulness} 
\label{sec:respdecidability}

The definition of orchestrators implies they have buffering capabilities.
The sort of buffer taken into account in \cite{Padovani10}, as well as by us, is made of a number of bi-directional buffers (where only a finite subset is actually non empty), one for each possible name.
A bi-directional buffer is actually made of two distinct buffers, one containing the messages
received from the client that have to be delivered to the server, and the other one containing the messages received from the server that should be delivered to the client. 

In \cite{Padovani10} orchestrators are restricted to have bounded buffering capabilities and
such a restriction is used in the proofs of several properties concerning contract orchestrators. 
In our setting we can eliminate that restriction, so allowing  more client/server pairs to be compliant, like for instance $\rec x \procdot a\Actdot x$ and $\rec x \procdot \Dual{b}\Actdot \Dual{a}\Actdot x$, and the example in the introduction.
We will now formalise the notion of buffer. 

\Comment{
\begin{definition}[Buffers]
A bi-directional buffer $\Buf$ is a set of the form $\Set{ \numberused{c_a}{a}{s_a}\mid a\in\Names } $ where,
for any $a\in\Names$, $c_a,s_a\in\mathbb{Z}$. 
The $c_a$ in $\numberused{c_a}{a}{s_a}$ represents the number of $a$'s elements in the part of the buffer containing elements sent by the client in order to be delivered to the server. 
The $s_a$ in $\numberused{c_a}{a}{s_a}$ represents the number of $a$'s elements in the part of the buffer containing elements sent by the server in order to be delivered to the client. We define: $\emptyBuf=\Set{ ^{0}\!a^{0}\mid a\in\Names } $ and\\
\centerline{$\begin{array}{rcl}
\Bufcs{a}{+}&=&(\Buf\setminus\Set{ \numberused{c_a}{a}{s_a} } )\cup\Set{ \numberused{c_a+1}{a}{s_a} } \\
\Bufcs{a}{-}&=&(\Buf\setminus\Set{ \numberused{c_a}{a}{s_a} } )\cup\Set{ \numberused{c_a-1}{a}{s_a} } 
\end{array}$
\hspace{10mm}
$\begin{array}{rcl}
\Bufsc{a}{+}&=&(\Buf\setminus\Set{ \numberused{c_a}{a}{s_a} } )\cup\Set{ \numberused{c_a}{a}{s_a+1} } \\
\Bufsc{a}{-}&=&(\Buf\setminus\Set{ \numberused{c_a}{a}{s_a} } )\cup\Set{ \numberused{c_a}{a}{s_a-1} } 
\end{array}$}

\noindent
Moreover, we denote by $\numberstoc |\Buf|_a$ the number of $a$'s in the server-to-client part of the buffer, i.e.~$|\Buf|_a = s_a$ and similarly for the client-to-server, i.e. $ _a|\Buf| = c_a$.

The state of a buffer
$\Buf$ after an orchestration action $ \mu$ is denoted by $\Buf\mu$, and is defined by\\
\centerline{$\begin{array}{rcl}
\Buf\orchAct{\alpha}{\Dual{\alpha}}&=&\Buf \\
\Buf\orchAct{\Dual{a}}{\varepsilon}&=&\Bufcs{a}{-} \\
\end{array}$
\hspace{14mm}
$\begin{array}{rcl}
\end{array}$
\hspace{14mm}
$\begin{array}{rcl}
\Buf\orchAct{\varepsilon}{\Dual{a}}&=&\Bufsc{a}{-} \\
\Buf\orchAct{\Dual{a}}{\varepsilon}&=&\Bufcs{a}{-} \\
\Buf\orchAct{\varepsilon}{a}&=&\Bufsc{a}{+}
\end{array}$}

\noindent
By $\Buf\vec{\mu}$ we denote the buffer $\Buf$ after the (finite) sequence $\vec{\mu}$ of orchestration actions. 
\end{definition}

}

\begin{definition}[Buffers]
 \begin{enumerate}
 
 \item
A bi-directional buffer $\Buf$ is a set of the form $\Set{ \numberused{c_a}{a}{s_a}\mid a\in\Names } $ where,
for any $a\in\Names$, $c_a,s_a\in\mathbb{Z}$. 
The $c_a$ in $\numberused{c_a}{a}{s_a}$ represents the number of $a$'s in the part of the buffer containing messages 
sent by the client  to the server. 
The $s_a$ in $\numberused{c_a}{a}{s_a}$ represents the number of $a$'s  in the part of the buffer containing messages sent by the server to the client. \Comment{
We define: $\EmptyBuf=\Set{ ^{0}\!a^{0}\mid a\in\Names } $ and
 \[ \begin{array}{rcl}
\Bufcs{a}{+}&=&(\Buf\setminus\Set{ \numberused{c_a}{a}{s_a} } )\cup\Set{ \numberused{c_a+1}{a}{s_a} } \\
\Bufcs{a}{-}&=&(\Buf\setminus\Set{ \numberused{c_a}{a}{s_a} } )\cup\Set{ \numberused{c_a-1}{a}{s_a} } 
 \end{array} 
\hspace{10mm}
 \begin{array}{rcl}
\Bufsc{a}{+}&=&(\Buf\setminus\Set{ \numberused{c_a}{a}{s_a} } )\cup\Set{ \numberused{c_a}{a}{s_a+1} } \\
\Bufsc{a}{-}&=&(\Buf\setminus\Set{ \numberused{c_a}{a}{s_a} } )\cup\Set{ \numberused{c_a}{a}{s_a-1} } 
\end{array} \]

\item
We denote by $\numberstoc |\Buf|_a$ the number of $a$'s in the server-to-client part of the buffer, i.e.~$|\Buf|_a = s_a$ and similarly for the client-to-server part, i.e.~$ {_a} |\Buf| = c_a$.
}

\item
We define: $\EmptyBuf=\Set{ ^{0}\!a^{0}\mid a\in\Names } $ and
\[ \begin{array}{rcl}
\Bufcs{a}{+}&=&(\Buf\setminus\Set{ \numberused{c_a}{a}{s_a} } )\cup\Set{ \numberused{c_a+1}{a}{s_a} } \\
\Bufcs{a}{-}&=&(\Buf\setminus\Set{ \numberused{c_a}{a}{s_a} } )\cup\Set{ \numberused{c_a-1}{a}{s_a} } \\
 \end{array} 
\hspace{10mm}
 \begin{array}{rcl}
\Bufsc{a}{+}&=&(\Buf\setminus\Set{ \numberused{c_a}{a}{s_a} } )\cup\Set{ \numberused{c_a}{a}{s_a+1} } \\
\Bufsc{a}{-}&=&(\Buf\setminus\Set{ \numberused{c_a}{a}{s_a} } )\cup\Set{ \numberused{c_a}{a}{s_a-1} } 
\end{array} \] 

\item
We denote by $\numberstoc |\Buf|_a$ the number of $a$'s in the server-to-client part of the buffer, i.e.~$|\Buf|_a = s_a$ and similarly for the client-to-server  part, i.e.~$\numberctos _a|\Buf| = c_a$. 

\item
The state of a buffer
$\Buf$ after an orchestration action $ \mu$ will
be denoted by $\Buf\mu$, defined by
 \[ \begin{array}{rcl}
\Buf\orchAct{\Dual{a}}{\varepsilon}&=&\Bufcs{a}{-} \\
\Buf\orchAct{a}{\varepsilon}&=&\Bufcs{a}{+} \\
 \end{array} 
 \hspace{14mm}
\begin{array}{rcl}
\Buf\orchAct{\alpha}{\Dual{\alpha}}&=&\Buf \\
\end{array}
\hspace{14mm}
 \begin{array}{rcl}
\Buf\orchAct{\varepsilon}{\Dual{a}}&=&\Bufsc{a}{-} \\
\Buf\orchAct{\varepsilon}{a}&=&\Bufsc{a}{+}
 \end{array} \] 

\item
By $\Buf\vec{\mu}$ we denote the buffer $\Buf$ after the sequence $\vec{\mu}$ of orchestration actions. 

 \end{enumerate}
 \end{definition}



In Definition~\ref{def:disstrictcompl} we considered the relation $\complyPds$, which we have studied so far.
This is however much weaker than
expected, and it is time to face the issue. 
Consider the simple orchestrated system
 \centerline{$ \begin{array}{rl}
\Dual{a}\Actdot \Dual{b} \pf{f} a\Actdot c\Actdot d & \textrm{where } f = \orchAct{a}{\Dual{a}}\Actdot \orchAct{b}{\varepsilon}\Actdot \stopf. 
 \end{array} $}

It is easy to check that $f: \Dual{a}\Actdot \Dual{b}\complyPds a\Actdot c\Actdot d$ since 
$f$ is strict for the given client/server pair and \linebreak $\Dual{a}\Actdot \Dual{b} \pf{f} a\Actdot c\Actdot d \Lts{\vec{\mu}} \stopA \pf{\stopf} c\Actdot d \notlts {} $, where $\vec{\mu}= \orchAct{a}{\Dual{a}}\orchAct{b}{\varepsilon}$. 
It is definitely true that all the client's ``requests'' have been satisfied, but not all by the server!
The action $\Dual{b}$ of the client has been taken care of exclusively by the orchestrator, which in that case has not acted simply as a mediator, but has effectively participated to the completion of the client's requests. 

So, in order to strengthen Definition~\ref{def:disstrictcompl} (\ref{def:disstrictcompl-i}), in case $\rho' \pf{f'} \sigma' \notlts {}$, we have to impose some conditions on the client-to-server buffer associated to $f'$; in particular, that it should be empty.
Of course, a similar condition must hold also for infinite interactions; this implies that in an infinite interaction, for any possible name, say $a$, used by the orchestrator, the latter cannot indefinitely perform input actions for $a$ from the client (even if interspersed with actions for other names) without ever delivering an $a$ to the server.
We must therefore forbid a client like $\rec x \procdot \Dual{a}\Actdot \Dual{c}\Actdot x$ to be compliant with the server $\rec x\procdot c\Actdot x$ by means of the orchestrator $\rec x \procdot\orchAct{a}{\varepsilon}\Actdot \orchAct{c}{\Dual{c}}\Actdot x$.
Orchestrated finite and infinite interaction sequences which do not correspond to unwanted situations like those just sketched will be called {\bf client-respectful}.

Even if the notion of compliance enforces the sense of the bias towards the client (any client request must be eventually satisfied by the server), some conditions need to be imposed on
the part of interactions on behalf of the server. In fact, we wish to prevent a server to be compliant with a client by means of an orchestrator that, from a certain moment on,  interacts infinitely many times with the server only, like in the orchestrated system\\
 \centerline{ $ \begin{array}{rl}
\Dual{a}\Actdot b \pb{f} \rec x \procdot \Dual{c}\Actdot \Dual{b}\Actdot  x
& \textrm{where } f=\orchAct{a}{\varepsilon} \Actdot\rec x \procdot \orchAct{\varepsilon}{c} \Actdot  \orchAct{\varepsilon}{b} \Actdot x
 \end{array} $}
We wish to prevent this kind of infinite interaction that we dub
{\bf definitely server-inputted}. Notice that, however, we can permit interactions in which the orchestrator can perform the input of some $a$ from the server infinitely many times, without ever performing an output of $a$ to the client, like it happens for {\tt wind} in the example in the introduction.

We observe that the problem -- whether an orchestrator will ever engage in any of the aforementioned pathological interactions --
might well be undecidable for contracts in general; indeed, it shares similarities with, for example, termination of two-counter machines \cite{countMachines}. 
However, we stress that we are in the restricted setting of session contracts, which suffices to make such properties decidable.

Among the properties we have to take care of, one is that in an interaction sequence there cannot exist an orchestrator action removing an element from an empty buffer, i.e.~a {\bf sound} sequence never sends an element $a$ to a server or to a client if the $a$ has not been previously received. 
We call the sum of all the above properties {\bf respectfulness}.

\begin{definition}\label{def:lr-restriction}
Given $\vec{\mu}\in\OrchAct^*\cup\OrchAct^\infty$, we define $\restrl{\vec{\mu}}{a}$, its left-restriction to a name $a$, as follows ($\lambda$ is the empty sequence):
 \[ \begin{array}{lcll}
\restrl{\lambda}{a} &=& \lambda, \\
\restrl{(\mu\vec{\mu'})}{a} &=& \mu\,\restrl{\vec{\mu'}}{a}, &\mbox{ if }  \mu\in\Set{
\orchAct{\varepsilon}{\Dual{a}}, \orchAct{a}{\varepsilon}}, \\
\restrl{(\mu\vec{\mu'})}{a} &=& \restrl{\vec{\mu'}}{a}& \mbox{ otherwise. }
\end{array} \]
\end{definition}


\begin{definition}[Respectful sequences and orchestrators]\label{def:respSeq} 
Let $\vec{\mu}\in\OrchAct^*\cup\OrchAct^\infty$ and 
$\mu\in\OrchAct$.
\begin{enumerate}[a)]
\item
Given $S\subseteq\OrchAct$, we say $\vec{\mu}$  to be {\em definitely-$S$}  whenever: 
 \[
\exists k ~ \forall m\geq k \Pred[ \textrm{the } m\textrm{-th element of $\vec{\mu}$ belongs to }S ];
 \]
For sets that are singletons  we write `definitely-$\mu$' instead of `definitely-$\Set{\mu}$.' 
\item
We say $\vec{\mu}$ to be a {\em sound sequence} whenever: 
 \[
\forall a \in \Names ~ \forall n\leq|\vec{\mu}| \Pred [ \leftinseq _a|\EmptyBuf\mu_1\cdots\mu_n|\geq 0 ~\mbox{ and }~ \rightinseq |\EmptyBuf\mu_1\cdots\mu_n|_a \geq 0 ]
 \]

\item \label{def:respSeq-iii}
We say $\vec{\mu}$ to be {\em client-respectful} sequence whenever, for any $a \in \Names$:
 \[ \begin{array}{ccc}
\restrl{\vec{\mu}}{a} \mbox{ is finite} \mbox{ and } \leftinseq _a |\EmptyBuf \vec{\mu}| = 0 
	& \mbox{ or  } & \restrl{\vec{\mu}}{a} \mbox{ is infinite and non-definitely-}\orchAct{a}{\varepsilon}
 \end{array} \]

\item
We say $\vec{\mu}$ to be {\em non-definitely server--inputted} whenever:
\[ \vec{\mu} \textrm{ is  infinite$\quad\implies\quad \vec{\mu}$  is non-definitely-}\Set{\orchAct{\varepsilon}{a} \mid a\in\Names}
 \]

\item\label{def:respSeq-v}
We say $\vec{\mu}$ to be {\em respectful} whenever 
$\vec{\mu}$ is sound, client-respectful and non-definitely server-inputted.
\item
We say that an orchestrator $f$ is {\em respectful} whenever every $\vec{\mu}\in \MaxTrace{f}$
is so.
\end{enumerate}
\end{definition}

We will look now at a few examples in order to get a better intuition about the above definition. 

 \begin{example}
 \begin{itemize}
 \item 
The finite sequence $\orchAct{a}{\varepsilon}\Actdot \orchAct{\varepsilon}{\Dual{b}}\Actdot \orchAct{\varepsilon}{\Dual{a}}$ is not respectful
since it is not sound. In fact, for the name $b$, we have that
$\rightinseq |\EmptyBuf\Actdot \orchAct{a}{\varepsilon}\Actdot \orchAct{\Dual{b}}{\varepsilon}|_b = -1 < 0$.

 \item
The sequence 
$\orchAct{a}{\varepsilon}\Actdot \orchAct{b}{\varepsilon}\Actdot \orchAct{\varepsilon}{\Dual{a}}$ instead, is sound, but nonetheless it is not client-respectful, since it is not infinite and for the name $b$ we have
$\leftinseq _b |\EmptyBuf \orchAct{a}{\varepsilon}\Actdot \orchAct{b}{\varepsilon}\Actdot \orchAct{\varepsilon}{\Dual{a}}| =1\neq 0 $. 

 \item
The orchestrator $f = \orchAct{c}{\Dual{c}}\Actdot \rec x \procdot (\orchAct{\Dual{a}}{a} \vee\orchAct{c}{\varepsilon}\Actdot \orchAct{b}{\Dual{b}}\Actdot x)$ is not respectful since it is not client-respectful. In fact, for the sequence
$\vec{\mu}=\orchAct{c}{\Dual{c}}\Actdot \orchAct{c}{\varepsilon}\Actdot \orchAct{b}{\Dual{b}}\Actdot \orchAct{c}{\varepsilon}\Actdot \orchAct{b}{\Dual{b}} \cdots \in \MaxTrace{f}$ and the name $c$, we have that  $\restrl{\vec{\mu}}{c}$ is infinite and $\restrl{\vec{\mu}}{c}=\orchAct{c}{\varepsilon}\Actdot \orchAct{c}{\varepsilon}\Actdot \orchAct{c}{\varepsilon}\cdots $ is definitely-$\orchAct{c}{\varepsilon}$. In fact, from the very first element on it is 
made of $\orchAct{c}{\varepsilon}$ actions. 

 \item
The orchestrator $f = \orchAct{c}{\Dual{c}}\Actdot \rec x \procdot (\orchAct{\Dual{a}}{a} \vee\orchAct{\varepsilon}{b}\Actdot \orchAct{\varepsilon}{c}\Actdot x)$ is not respectful since it is not definitely server-inputted. In fact, the infinite sequence $\vec{\mu}=\orchAct{c}{\Dual{c}}\Actdot \orchAct{\varepsilon}{b}\Actdot \orchAct{\varepsilon}{c}\Actdot \orchAct{\varepsilon}{b}\Actdot \orchAct{\varepsilon}{c} \cdots\in \MaxTrace{f}$ is  definitely-$\Set{\orchAct{\varepsilon}{a} \mid a\in\Names}$.  
The orchestrator {\tt f} in the introduction, instead, is non-definitely server-inputted, and also respectful, as a matter of fact.
 \end{itemize}
 \end{example}

\begin{remark}\label{rem:impl}
{
By Definition~\ref{def:lr-restriction}, the sequence $\restrl{\vec{\mu}}{a}$ in Definition~\ref{def:respSeq}(\ref{def:respSeq-iii}) cannot contain synchronous orchestration actions like $\orchAct{a}{\Dual{a}}$. Hence, for example, the orchestrator  
$g = \orchAct{a}{\varepsilon}\Actdot \rec x \procdot \orchAct{a}{\Dual{a}}\Actdot x$ is not client-respectful, and so it is not respectful at all. This is because the first $a$ coming from the client will never be delivered to the server since any subsequent output $\Dual{a}$ will be paired with a further input of $a$. 
This might be irrelevant when distinct occurrences of the same message are indistinguishable, but in general the number of input-output actions matters.

On the other hand forcing the orchestrator to immediately forward a message is a desirable capability, which would be definitely lost by equating
$\orchAct{a}{\varepsilon}\Actdot \orchAct{\varepsilon}{\Dual{a}}$ to $\orchAct{a}{\Dual{a}}$, and by ruling out the latter.
}
\end{remark}

We can now properly define the full notion of compliance and characterise it.
\begin{definition}[Orchestrated Session Compliance]\label{def:orchCompl}
\begin{enumerate}[i)] 
\item\label{def:orchCompl-a}
We say that a client $\rho$ is compliant with a server $\sigma$ through the orchestration of $f$, and denote this by $f: \rho\complyP \sigma$, whenever 
\begin{enumerate}[a)] 
\item
$\rho \pf{f} \sigma \Lts{\vec{\mu}} \rho' \pf{f'} \sigma' \notlts{} $ ~~ implies~~ $\rho'=\mbox{\em $\stopA$} $~ and~ $\vec{\mu}$ is respectful, ~~ and
\item
$\rho \pf{f} \sigma \Lts{\vec{\mu}}$ with  $\vec{\mu}\in\OrchAct^\infty $
~~implies~~ $ \vec{\mu}$ is respectful.
\end{enumerate}
\item
We write $\rho\complyP\sigma$ whenever there exists an orchestrator $f$ such that $f: \rho\complyP \sigma$.
\end{enumerate}
\end{definition}
Notice that we cannot define orchestrated compliance by simply imposing $f$ to be respectful 
in Def.~\ref{def:orchCompl}(\ref{def:orchCompl-a}), since that would prevent the possibility of
$\Dual{a}$ be compliant with $a$ through the mediation of the orchestrator $\orchAct{a}{\Dual{a}}\vee\orchAct{\varepsilon}{\Dual{b}}$. This orchestrator is not respectful, but
its sequences of actions in any possible orchestration between $\Dual{a}$ and  $a$ are respectful.

We can show that, if compliance could be obtained by means of a non-respectful orchestrator, it is always 
possible to get it through a respectful one. Besides, we can show the correspondence between
$\complyP$ and $\complyPds$.

\begin{proposition}\label{prop:charcompl}
\begin{enumerate}[i)]
\item 
$ \begin{array}{@{}rcl}
f:\rho\complyP \sigma &\implies& \exists f' \Pred[f':\rho\complyP \sigma \mbox{ such that $f'$ is $\rho \Strictdot \sigma$ strict}].
\end{array} $
\item
\label{prop:charcompl-i}
$ \begin{array}{@{}rcl}
f:\rho\complyP \sigma\textit{ and $f$ is $\rho \Strictdot \sigma$ strict} &\Iff& f: \rho\complyPds \sigma\textit{ and $f$ is respectful.}
\end{array} $
\item
\label{prop:charcompl-ii}
$ \begin{array}{@{}rcl}
\rho\complyP \sigma &\Iff& \exists f \Pred[ f: \rho\complyPds \sigma\mbox{~where $f$ is respectful}]. 
\end{array} $
\end{enumerate}
\end{proposition}

In order to show decidability, we provide a characterisation of respectfulness  based on the notion of buffer-aware trees and its related labelings below.

\begin{definition}[Buffer-aware trees of $f$] 
\begin{enumerate}[a)]
\item
Let $a\in\Names$. We define the {\em buffer-aware $a$-tree of an orchestrator $f$}, denoted by
$\csatree{a}{f}$, as the tree defined by induction in Figure~\ref{bufferawaretree}. The edges of the tree have a left- and a right-weight denoting,
respectively, the increment of the client-to-server and of the server-to-client buffer for the name  `a' caused by the orchestration actions performed by $f$. 

Given an edge $e$ of a buffer-aware $a$-tree $t$, we denote is left (resp.~right) weight by
${\sf lw}^t(e)$ (resp.~${\sf rw}^t(e)$).
\item
We define the {\em buffer-aware $*$-tree of an orchestrator $f$}, denoted by
$\csatree{*}{f}$, as the tree with the same nodes and edges as any $\csatree{a}{f}$, but such that the left (resp.~right) weight of an edge $e$ is $\sum_{a\in\Names}{\sf lw}^{\csatree{a}{f}}(e)$ (resp.~$\sum_{a\in\Names}{\sf rw}^{\csatree{a}{f}}(e)\ )$.
\end{enumerate}
\end{definition}

Note that the left and right weights of the edges of a  buffer-aware $*$-tree of an orchestrator $f$
are either $0$, $-1$, or $+1$.

 \begin{figure}[t]
 \hrule
 \vspace*{2mm}
 \[ \begin{array}{rcl@{\hspace*{2cm}}rcl}
\csatree{a}{\stopf} &=& \stopf 
&
\csatree{a}{x} & = & x 
	\\ [2mm]
\csatree{a}{\orchAct{\varepsilon}{\Dual{a}}\Actdot f'}
	&= &\begin{array}{c}
		\circ \\[-1mm]
		 ~~ \hspace{-2pt}\mbox{\footnotesize $0$}\mid \mbox{\footnotesize -$1$} \\ 
		\csatree{a}{f'}
	 \end{array}
&
\csatree{a}{\orchAct{a}{\varepsilon}\Actdot f'}
	&= &	\begin{array}{c}
		\circ \\[-1mm]
		\hspace{-2.3mm} \mbox{\footnotesize $+1$} \mid \mbox{\footnotesize $0$} \\ 
		\csatree{a}{f'}
	 \end{array}
	\\ [8mm]
\csatree{a}{\orchAct{\Dual{a}}{\varepsilon}\Actdot f'}
	&= &	\begin{array}{c}
		\circ \\[-1mm]
		\hspace{-1mm}\mbox{\footnotesize -$1$} \mid \mbox{\footnotesize $0$} \\ 
		\csatree{a}{f'}
	 \end{array}
&
\csatree{a}{\orchAct{\varepsilon}{a}\Actdot f'}
	&= &	\begin{array}{c}
		\circ \\[-1mm]
		 ~~~\mbox{\footnotesize $0$} \mid \mbox{\footnotesize $+1$} \\ 
		\csatree{a}{f'}
	 \end{array}
	\\ [8mm]
\csatree{a}{\mu\Actdot f'} & = &
\multicolumn{4}{l}{
	\begin{array}{c}
		\circ \\[-1mm]
		\mbox{\footnotesize $0$} \mid \mbox{\footnotesize $0$} \\ 
		\csatree{a}{f'} 
	 \end{array} 
\textrm{if } \mu\not\in \Set{ \orchAct{\varepsilon}{\Dual{a}},\orchAct{a}{\varepsilon}, \orchAct{\Dual{a}}{\varepsilon}, \orchAct{\varepsilon}{a} }
}
	\\ [8mm]
\csatree{a}{f_1\vee\ldots\vee f_n}
	&=& \begin{array}{c}
		\circ \\[-1mm]
		\diagup \ldots \diagdown\\
		\csatree{a}{f_1}\ldots\csatree{a}{f_n}
	 \end{array}
&
\csatree{a}{\rec x \procdot f'}
	&=& \begin{array}{c}
		\rec x \procdot \\[-1mm]
		 \mid\\
		\csatree{a}{f'}
	 \end{array}
 \end{array} \]
\caption{Buffer-aware $a$-tree}\label{bufferawaretree} 
 \vspace*{2mm}
 \hrule
 \end{figure}


\begin{definition}[Buffer-labelling of $\csatree{a}{f}$] 
We define the {\em buffer-labelling of} $\csatree{a}{f}$ by labelling its nodes with left and right labels as follows: 
given a node $N$ and the path $P$ in $\csatree{a}{f}$ from the root to $N$, we left-label $N$ with the sum of all the left-weights of the edges in $P$, whereas we right-label $N$ with the sum of all the right-weights of the edges in $P$.
 \end{definition}
We now provide characterisations for the properties defining respectfulness.
\begin{definition}[Sound buffer-labelling]\label{def:sound-blab} 
The buffer-labelling of $\csatree{a}{f}$ is {\em sound} whenever
\begin{enumerate}[a)]
\item \label{sb1}
there is no negative left-label and no negative right-label and
\item \label{sb2}
for any leaf $x$ and corresponding $\rec x\procdot$ node, if $k$ is the left (resp.~right) label of $x$ and $h$ is the left (resp.~right) label of $\rec x\procdot$, then: $k-h\geq 0$.
 \end{enumerate}
 \end{definition}

\begin{proposition}~~~ $f$ is sound  ~~~$\Leftrightarrow$~~~ for any $a\in\Names$, the buffer-labelling of $\csatree{a}{f}$ is sound.
 \end{proposition}
 
\begin{proof}
($\Leftarrow$) By the labelling, it is impossible to get a non-client-respectful sequence out of $f$.\\
($\Rightarrow$) By contraposition; assume that for a name
$b\in\Names$, the buffer-labelling of $\csatree{b}{f}$ be unsound. Then we have two cases to consider:
\begin{enumerate}[(a)]
\item
There is a negative label. 
We then get immediately an unsound sequence.
\item
There exists a leaf $x$ and its corresponding $\rec x\procdot$ node, where $k$ is the left(or right-)-label of $x$ and $h$ is the left-(or right-)label of $\rec x \procdot $, s.t. ~$k-h < 0$. 
It is immediate to get an unsound sequence.
 \end{enumerate}
\vspace{-7mm}
 \end{proof}

We say that a node {\em gets to} $\stopf$ whenever its subtree contains a $\stopf$ node.

\begin{definition}[Client-respectful buffer-labelling]\label{def:client-res-plab} 
The buffer-labelling of $\csatree{a}{f}$ is {\em client-respectful} whenever 
\begin{enumerate}[a)]
\item\label{crb1}
any $\stopf$ node is left-labelled with $0$;
\item\label{crb2}
for any leaf $x$ and 
corresponding node of its binder $\rec x \procdot $, if $k$ is the left-label of $x$ and $h$ is the left-label of $\rec x \procdot $, then
 \begin{enumerate}[1)]
\item\label{crb1A}
if the $\rec x \procdot $ node gets to $\stopf$, then $h=k$;
\item\label{crb2A}
otherwise, if all the left-labels of the edges from $x$ to $\rec x \procdot $ are $0$ then $h=0$;
 \end{enumerate}
\item\label{crb3}
for any path from a leaf $x$ to its
corresponding $\rec x \procdot $ node, either no edge is right-weighted with $+1$ or there is at least an edge with right-weight $-1$.

 \end{enumerate}
 \end{definition}

\begin{proposition}\hfill\\
\centerline{$f$ is client-respectful  ~~~$\Leftrightarrow$~~~ for any $a\in\Names$, the buffer-labelling of $\csatree{a}{f}$ is client-respectful.}
 \end{proposition}
\begin{proof}
 ($\Leftarrow$)
By the labeling rule it is impossible to get a non client-respectful sequence out of $f$. For finite sequences this impossibility is guaranteed by clauses (\ref{crb1}) and (\ref{crb2}) of
Definition~\ref{def:client-res-plab}, for infinite ones by clause (\ref{crb3}).\\
($\Rightarrow$) 
By contraposition; assume that for a name $b\in\Names$, the buffer-labelling of $\csatree{b}{f}$ be non-client-respectful. We consider the four 
possible cases:

\begin{enumerate}

\item A label of a $\stopf$ leaf is not $0$. 
In that case we immediately get a finite sequence 
out of $f$ which is non-client-respectful.
\item There is a node $x$ labelled with $k$ and its corresponding node $\rec x \procdot $ gets to $\stopf$ and it is labelled with $h$, with $k\neq h$. 
Then the sequence out of $f$ corresponding to going to $\rec x \procdot $, then from node $\rec x \procdot $ to $x$ a non negative number of times $n$ and finally to the $\stopf$ node
cannot be client-respectful, since at the end the client-to-server buffer for $b$ would have $n*(h-k)$ elements in it. 

\item 
There is a node $x$ labelled with $k$, its corresponding node $\rec x \procdot $ does not get to $\stopf$, all the left-labels of the edges from $x$ to $\rec x \procdot $ are $0$ and $h=0$. 
In that case the trace $\vec{\mu}$ corresponding to the infinite path starting from the root and then keeping indefinitely on passing 
through $\rec x \procdot $ and $x$ is such that $\restrl{\vec{\mu}}{b}$ is finite and $|\EmptyBuf(\restrl{\vec{\mu}}{b})|\neq 0$.

\item
 there exists a path from a leaf $x$ to its corresponding $\rec x \procdot $ such that there are some right-weighted edges right-weighted with $+1$ and no edge with right-weight $-1$.
Then it is immediate to get an infinite definitely server-inputted sequence out of $f$ which is
definitely-$\orchAct{b}{\varepsilon}$ and hence not client-respectful.
 \end{enumerate}
\vspace{-6mm}
 \end{proof}

\begin{definition}[Non definitely server-inputted $*$-tree] 
Given an orchestrator $f$, its $*$-tree $\csatree{*}{f}$ is {\em\ non-definitely server-inputted} whenever, for any path from a leaf $x$ to its
corresponding $\rec x \procdot $ node, either no edge is right-weighted with $+1$ or there is at least an edge with right-weight $-1$.
 \end{definition}

\begin{proposition} ~~~ $f$ is not definitely server-inputted  ~~~$\Leftrightarrow$~~~$\csatree{*}{f}$ is non-definitely server-inputted.
 \end{proposition}

\begin{proof} 
($\Leftarrow$)
By the labelling it is impossible to get a definitely server-inputted sequence out of $f$.\\
($\Rightarrow$)
By contraposition; assume that $\csatree{*}{f}$ be definitely server-inputted.
So there exists a path from a leaf $x$ to its corresponding $\rec x \procdot $ such that there are some right-weighted edges right-weighted with $+1$ and no edge with right-weight $-1$.
Then it is immediate to get a definitely server-inputted sequence out of $f$.
 \end{proof}
\begin{theorem}\label{thr:respdecidability}
Orchestrator respectfulness is decidable.
 \end{theorem}

From the above result and from decidability of $\complyPds$ (
Corollary~\ref{cor:decidable-complPds}) we can get decidability of  $\complyP$.
The algorithm to decide whether $\rho\complyP \sigma$ will first compute ${\cal F} = \Synth(\emptyset,\rho,\sigma)$; then if ${\cal F} \neq \emptyset$ 
it suffices to check whether there is a strict and respectful $f\in {\cal F}$, which is a decidable problem by the above.

\begin{theorem}
Given $\rho$ and $\sigma$, it is decidable whether $\rho\complyP\sigma$.
 \end{theorem}

We conclude by observing that in \cite{Padovani10} the lack of unbounded buffering capabilities  prevents orchestrators to be used to ensure client compliance with a server that might send an unbounded number of unnecessary outputs.
To let such sort of interaction possible, in \cite{BdL14b}  
the notion of ${\tt skp}$-compliance (dubbed $\complyG$) was investigated for session contracts, where a client is compliant with a server whenever all its requests can be satisfied thanks to the possibility of discarding a (possibly unbounded) number of unnecessary server outputs.
Interactions of this sort can actually be carried out by means of our session orchestrators, since it is possible to prove that $ \begin{array}{rcl} \rho\complyG\sigma &\mbox{implies}& \rho\complyPds \sigma. \end{array}$ 
In the example in the introduction, in fact, the {\tt wind} information is unbounded and ``discarded'' by the orchestrator. 

\section{Related and future work}\label{rwc}

The notion of compliance naturally induces a substitutability relation on servers that may be used for implementing contract-based query engines (see \cite{Padovani10} for a detailed discussion). Hence it seems
worthwhile to investigate the
session sub-contract relation induced by our orchestrated compliance on session contracts.
Whereas server substitutability is at the core of the results in \cite{Padovani10}, we deem it relevant to investigate also {\em client substitutability}, in the style of what was done in \cite{BdL10,BdL13} for session contract and in \cite{BH13c} for the more general notion of contract. 

An approach to the formal description of service contracts in terms of automata
has been recently developed in \cite{BDF14}. 
The notion of {\em contract automaton} is related to that of {\em contract} as well as of {\em session contract}. 
Besides, the notion of {\em contract agreement} in \cite{BDF14} somewhat resembles that of {\em compliance}. 
In the framework of that paper, orchestrators are synthesised  to enforce contract composition to adhere to the requirements for contract agreement. 
Even if the authors of \cite{BDF14} work on the overall satisfaction in a multiparty composition of principals, it is definitely worthwhile, as a future investigation, to study the relation between the notion of orchestration, as developed in \cite{Padovani10} and in the present paper, and the approach of \cite{BDF14}, which in turn has been related in \cite{BDFT14} to the model of choreography of communicating finite state machines (CFMS) \cite{BZ83}.
For what concerns {\em session contracts} in particular, the investigation of the correspondence with the above mentioned formalisms could start from the result concerning the correspondence of {\em binary} session types with a particular two-communicating-machines subclass (see \cite{DY12} for references). 
Such a correspondence between session types and communicating machines has been pushed further to the multiparty setting in \cite{DY12}. 

Many properties of the model of CFSM which are untractable ceases to be so when Bags, instead of - or together with - FIFO queues are taken into account  \cite{CHF14}. The similarity of 
contracts and session contracts with the CFSM model suggests to investigate the use of bags for session-contract interactions  to reduce decidability problems
in our context to problems in the CFSM model with bags. What does a bag correspond to in our context is however not immediate to device. In fact, by putting a bag in between $\Dual{a}.\Dual{b}$ and $a+b$ would result in a number of possible non-deterministic evolutions of the system: as soon as $a$ is in the bag, it could be used as input for the server; or, in case both $a$ and $b$ get into the bag, the server could non-deterministically choose amongst them; etc. Such a behaviour of the system, however, goes far beyond the session setting we are in, where non-determinism is restricted to occur only inside the client and server. 


Session contracts have been also investigated in papers like \cite{BCP14,BSZ14concur} where,
overloading the name, they also have been dubbed {\em session types}. 
In \cite{BCP14} the authors establish a relation between session contracts and a model based on game-theoretic notions, showing that compliance corresponds to the existence of particular winning strategies. 
It should be interesting to investigate the meaning and role of the notion of orchestration in such a game-theoretical setting. 
\vspace{-4mm}
\paragraph{Acknowledgments.} 
We are grateful to the referees for their helpful and meaningful advices. The interaction with them has been pleasing and fruitful thanks to the forum tool provided by the workshop organisation.
We also wish to thank Mariangiola Dezani for her everlasting support.
\vspace{-4mm}

\bibliographystyle{eptcs}

\bibliography{session}

\end{document}